\documentclass[11pt,reqno]{amsart}
\usepackage{amsfonts,amsmath,amssymb}
\newtheorem{prop}{Proposition}

\newtheorem{cor}{Corollary}

\def\beq#1#2\eeq{%
        \begin{equation}%
        \label{#1}%
            #2%
        \end{equation}%
    }

\newcommand{\dee}{{\mathrm d}}

\usepackage{color}
\usepackage{graphicx}
\usepackage{epstopdf}

\theoremstyle{plain}
\newtheorem{theorem}{Theorem}

\theoremstyle{remark}

\theoremstyle{definition}

\title[Bianchi-VII)]{Quantum Bianchi-VII problem, Mathieu functions and arithmetic}

\author{A.P. Veselov}
\address{Department of Mathematical Sciences,
Loughborough University, Loughborough LE11 3TU}
\email{A.P.Veselov@lboro.ac.uk}

\author{Y. Ye}
\address{Department of Foundational Mathematics, Xi’an Jiaotong-Liverpool University, Suzhou, 215000, China}
\email{yiru.ye@xjtlu.edu.cn}

\begin{document}

\maketitle

\begin{abstract}
The geodesic problem on the compact threefolds with the Riemannian metric of Bianchi-VII$_0$ type is studied in both classical and quantum cases. We show that the problem is integrable and describe the eigenfunctions of the corresponding Laplace-Beltrami operators explicitly in terms of the Mathieu functions with parameter depending on the lattice values of some binary quadratic forms. We use the results from number theory to discuss the level spacing statistics in relation with the Berry-Tabor conjecture and compare the situation with Bianchi-VI$_0$ case (Sol-case in Thurston's classification) and with Bianchi-IX case, corresponding to the classical Euler top. 
 \end{abstract}

%\tableofcontents

\section{Introduction}

Let $M^3$ be a closed Riemannian 3-fold and consider the corresponding geodesic problem. It is natural to expect that the question of its integrability is connected with special geometric properties of $M^3.$ 

One such property is local homegeneity, meaning that any two points have isometric neighbourhoods. 
Under certain additional assumptions
Thurston classified all special geometries in dimension 3, showing that there are 8 of them: the Euclidean $E^3$, spherical $S^3$ and hyperbolic $\mathbb{H}^3$, the product
type $S^2 \times {\mathbb R}$ and $\mathbb{H}^2 \times {\mathbb R}$ and three
geometries related to the 3D Lie groups: $Nil$, $Sol$ and $\widetilde{SL(2,\mathbb R)}$, where the last group is the universal covering of $SL(2,\mathbb R)$ (see \cite{Th}).
The topology of the corresponding compact quotient threefolds is discussed in \cite{Scott}. 

On the other hand, Bianchi \cite{Bianchi} classified all the three-dimensional Lie groups, which provide the most natural class of locally homogeneous manifolds if we supply them with a left-invariant metric. The correspondence between Thurston's and Bianchi's classification was discussed in relation with the cosmological models in \cite{Fagundes,FIK} (see the next section).

The geodesic flows for the 3-folds with  $Nil$, $Sol$ and $SL(2,\mathbb R)$ Thurston's geometries (corresponding to Bianchi VI$_0$ and VIII types respectively) were studied in \cite{Butler, BT, BVY}, which revealed the coexistence of the chaos and integrability in these systems. The quantum version of the $Sol$-case was studied in \cite{BDV} in relation with quantum monodromy and arithmetic. The general Bianchi-IX case is well-known in classical mechanics as Euler top \cite{Arnold} and its quantum version was studied in \cite{KI} (see also \cite{GV}).

In this paper we will discuss the last remaining Bianchi-VII$_0$ case. 
The corresponding Lie group is group $E(2)=SO(2)\ltimes \mathbb R^2$ of orientation preserving isometries of Euclidean plane. This case is missing in Thurston's approach since this isometry group is not maximal and can be extended to the group $E(3)$, corresponding to the 3D Euclidean geometry. 

The classical system describing geodesics of left-invariant metric on $E(2)$ is completely integrable, as it was already shown by Bolsinov and Bao \cite{BB}. We will show that, in contrast to the $Sol$-case \cite{BT}, the integrability remains on the compact quotients $M^3_n=E(2)/G_n$ of $E(2)$ by a discrete crystallographic subgroup $G_n=\mathbb Z_n \ltimes \mathcal L \subset E(2)$, where $\mathbb Z_n$ is generated by an element $R \in SO(2)$, which is an order $n$ rotational symmetry of the corresponding 2D lattice $\mathcal L \subset \mathbb R^2$. 

There are 5 types of such threefolds corresponding to possible values of $n=1,2,3,4,6$ (see e.g. \cite{Coxeter}). 
They can be viewed as torus bundles over $S^1$ with gluing map given by an elliptic element $A \in SL(2, \mathbb Z).$ In that sense they are similar to the $Sol$-manifolds $M^3_A$ with gluing map  $A$ being a hyperbolic element of $SL(2, \mathbb Z)$ (see \cite{BDV}). However, we will see that the difference between these two cases is quite substantial. In particular, in contrast to the $Sol$-manifolds with  geodesic flow being non-integrable in analytic category and having positive topological entropy \cite{BDV,BT}, in the Bianchi-VII case we have completely integrable geodesic flow on $M^3_n$ with integrals polynomial in momenta and an explicit integration in terms of the elliptic functions.

We show that in the quantum case the spectral problem for the corresponding Laplace-Beltrami operator on $M^3_n$ can be explicitly solved in terms of the classical Mathieu functions \cite{WW}.
However, the properties of the spectrum depend on the order $n$ of the gluing map $R$.
We show that in the symmetric case $M^3_n$ with $n=3,4,6$ the spectrum has nice arithmetic properties similar to the $Sol$-case studied in \cite{BDV} and related to the classical problem of describing the values of certain integer quadratic forms $Q$ on two-dimensional integer lattice. 

In particular, for $n=4$ corresponding to the square lattice the parameters of the Mathieu equation and the corresponding eigenvalues depend on the values of the quadratic form $Q(x,y)=x^2+y^2, x,y \in \mathbb Z.$ The study of these values is a very classical part of number theory with famous contributions from Fermat, Euler, Gauss and Jacobi. A classical result due to E. Landau \cite{Landau} says that the number of integers up to a number $K$ represented by this form grows as
${K}/{\sqrt{\log K}}$.
This means that the spectrum of the corresponding symmetric Bianchi threefold $M^3_4$ is highly degenerate and the corresponding level space statistics does not satisfy the Berry-Tabor conjecture \cite{Berry-Tabor}.
Similar results hold for $n=3,6$ related to hexagonal lattice and the integer values of the quadratic form $Q(x,y)=x^2+xy+y^2.$
%The situation here is similar to the $Sol$-case \cite{BDV}, but there the quadratic form is indefinite and determined by the topology of $M_A^3$, while in our case it is determined by the lattice geometry.

For $n=1,2$ the situation is more delicate and depends on the geometry of the lattice and on the diophantine properties of the corresponding quadratic form $Q$.
In particular, from the deep results of Sarnak \cite{Sarnak} and Eskin, Margulis and Mozes \cite{Eskin} it follows that the pair correlation of the corresponding values does satisfy Berry-Tabor prediction on a set of full measure in the space of quadratic forms, in particular for $Q(x,y)=x^2+\beta y^2$ with diophantine $\beta.$
Our results suggest that the same should be true for the spectrum of the corresponding Bianchi-VII threefolds.

Finally, in the last section we make a comparison with the situation in Bianchi-VI and Bianchi-IX cases.

\section{Geodesic flow on Bianchi-VII$_0$ threefolds.}

In 1898 Bianchi \cite{Bianchi} completed the classification of 3D real Lie algebras by splitting them into 9 types: I-IX.
Since we are interested in compact quotients of the corresponding Lie groups, we must consider only the Lie algebras admitting bi-invariant volume form, which are called {\it unimodular}.
In Bianchi's list the unimodular Lie algebras are of types
\begin{itemize}
{\it 
\item I: abelian $\mathbb R^3$
\item II: nilpotent (Heisenberg) $nil$
\item VI$_0$: solvable 2D Poincare $sol \cong e(1,1)$
\item VII$_0$: solvable 2D Euclidean $e(2)$
\item VIII: simple $sl(2,\mathbb R)$
\item IX: simple $so(3).$
}
\end{itemize}
The curvature of the corresponding left-invariant metrics was studied in details by Milnor in \cite{Milnor}.

One can compare this list with Thurston's list of 8 special geometries in 3 dimensions \cite{Th}:
{\it Euclidean $E^3$, spherical $S^3$, hyperbolic $\mathbb{H}^3$, the product
type $S^2 \times {\mathbb R}$ and $\mathbb{H}^2 \times {\mathbb R}$ and three
geometries related to 3D Lie groups $Nil, \, Sol, \, SL(2,\mathbb R).$}

Recall that special geometry here means local homogeneity (any two points have isometric neighbourhoods),
existence of compact models and some maximality condition. Note that general Bianchi VII$_0$ and IX types are missing in Thurston's list since they do not satisfy the maximality condition having more symmetric versions: Euclidean $E^3$ and spherical $S^3$ respectively.

The integrability of the geodesic flow on 3D manifolds with Thurston's geometries of group types $Nil, \, Sol, \, SL(2,\mathbb R)$ (corresponding to Bianchi types II, VI$_0$, VIII) was studied by Butler \cite{Butler}, Bolsinov and Taimanov \cite{BT} and by Bolsinov and the authors \cite{BVY} respectively. The quantum $Sol$-case (or, equivalently, Bianchi-VI$_0$ case) was investigated by Bolsinov, Dullin and one of the authors in \cite{BDV}. 

The general Bianchi-IX$_0$ problem, corresponding to the geodesic flow on the group $SO(3)$ with left-invariant metric, is well known in classical mechanics as Euler top \cite{Arnold}. Its quantum version was studied at the very early age of quantum mechanics by Kramers and Ittmann \cite{KI}, who discovered close relations with the famous Lam\`e equation 
(see also \cite{GV}). In Thurston's list Bianchi-IX case is represented only by its most symmetric version $S^3$, which can be considered as the group $SU(2)$ with bi-invariant metric.

In this paper we will study the only remaining quantum Bianchi-VII$_0$ case, which has many parallels with Bianchi-VI$_0$ case, but deserves a special investigation.

%Note that the general quantum Bianchi-IX case, corresponding to geodesics on the group $SO(3)$ with left-invariant metrics, is known as quantum Euler top. It was studied already at the very early age of quantum mechanics by Kramers and Ittmann \cite{KI}, who discovered close relations with the famous Lam\`e equation (see also more recent work in this direction by Grosset and one of the authors \cite{GV}).
%{\color{red} We need to add some discussion here including the general class of metric of Bianchi-VII and some references ...}

We follow here the notations of Bolsinov and Bao \cite{BB}, where one can find also the relevant references about geodesic flows on 3 dimensional Lie groups.

The Lie algebras of Bianchi type VII have 3 generators $X_0, X_1, X_2$, satisfying the relations
$$[X_0,X_1]=a X_1-X_2, \,\, [X_0, X_2]=X_1+aX_2,\, [X_1, X_2]=0,$$
depending on a parameter $a.$ These Lie algebras are solvable, but for $a\neq 0$ are not unimodular.

When $a=0$ we have unimodular Bianchi subcase VII$_0$, corresponding to the Lie group $E(2)$ of isometries of the Euclidean plane.
The Lie algebra $e(2)$ is the semi-direct product of $\mathbb R^2$ by $\mathbb R$ with 3 generators $X_0, X_1, X_2$, satisfying the relations
\beq{VII0}
[X_0,X_1]=-X_2, \,\, [X_0, X_2]=X_1,\, [X_1,X_2]=0.
\eeq
The corresponding left-invariant vector fields are
$$
X_0=\partial_0,\, X_1=\cos q_0 \partial_1-\sin q_0 \partial_2,\, X_2=\sin q_0 \partial_1+\cos q_0 \partial_2,
$$
where $\partial_i=\frac{\partial}{\partial q_i}, \, i=0,1,2.$

The Hamiltonians of the geodesic flow corresponding to a left-invariant metric can be reduced (see \cite{BB}) to the form
\beq{Ham}
H=\frac{1}{2}(AX_0^2+BX_1^2+X_2^2),
\eeq
where $X_0, X_1, X_2$ are considered as linear functions on the dual space $e(2)^*.$
%In the simplest case of the square lattice we have the following Hamiltonian on the 3D torus $M^3=T^3$ with the angle coordinates $q_0, q_1, q_2$ defined modulo $2\pi$ (see \cite{BB}):
In the canonical coordinates $(p,q)$ we have
%$$
%H=\frac{1}{2}(Ap_0^2+B(p_1\cos q_0 -p_2 \sin q_0 )^2+( p_1\sin q_0+ p_2\cos q_0)^2)
%$$
%or, equivalently,
\beq{Ham}
H=\frac{1}{2}(Ap_0^2+p_1^2+p_2^2+(B-1)( p_1\cos q_0-p_2\sin q_0)^2),
\eeq
where $A,B>0$ are two real parameters.

The system is clearly completely integrable in Liouville sense with the commuting integrals $H$ and $F_1=p_1, F_2=p_2.$

Corresponding system can be easily integrated in quadratures (see \cite{BB}). Indeed, fixing the values of the integrals $F_1=p_1=c_1, \, F_2=p_2=c_2$ and assuming without loss of generality that $c_2=0$, we can write the equation of motion as
\beq{eqs}
\begin{cases}
\dot q_0 = Ap_0\\
\dot p_0 = (B-1)c_1^2 \sin q_0 \cos q_0=\frac{1}{2}(B-1)c_1^2 \sin 2q_0\\
\dot q_1 = c_1+(B-1)c_1 \cos^2 q_0=c_1 \sin^2 q_0+Bc_1 \cos^2 q_0\\
\dot q_2=-(B-1)c_1 \sin q_0 \cos q_0=-\frac{1}{2}(B-1)c_1 \sin 2q_0
\end{cases}
\eeq 
with $\dot p_1=\dot p_2=0.$ Note that the first two equations are simply the equations of motion of the mathematical pendulum with 
\beq{pend}
H_0=\frac{1}{2}p_0^2+C \cos 2q_0, \quad C=\frac{1}{4}(B-1)c_1^2,
\eeq
which can be integrated explicitly in terms of the Jacobi elliptic functions. Qualitatively, if the energy $H_0=h_0>|C|$ then we have rotations, while for $|h_0|<|C|$ the orbits are periodic. For the critical levels $h_0=\pm|C|$ we have separatrices and the stable equilibriums respectively (see Fig.1).

\begin{figure}[h]
  \includegraphics[width=120mm]{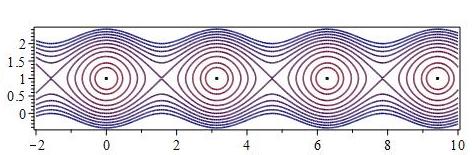}
  \caption{Phase portrait of mathematical pendulum}
\end{figure}

For known $(p_0(t), q_0(t))$ we can find $q_1(t)$ and $q_2(t)$ by integrating the last two equations, which also can be done explicitly in terms of the Jacobi elliptic functions. Note that since $B>0$ for positive $c_1$ $\dot q_1$ is always positive, so $q_1$ monotonically increases,
while $\dot q_2=-c_1^{-1}\dot p_0$, so $$q_2=-c_1^{-1}p_0+const$$ is always bounded. 

This describes the corresponding geodesics on the group $E(2).$ %(or better say, on its universal covering $\widetilde {E(2)}$).

Consider now the compact case $M^3_n=E(2)/G_n, \, G_n=\mathbb Z_n \ltimes \mathcal L \subset E(2)$, where $\mathbb Z_n$ is generated by matrix $R_n \in SO(2)$, defining the rotational symmetry of the corresponding 2D lattice $\mathcal L \subset \mathbb R^2$. More precisely, in types 1 and 2 with generic lattice $R_1=I$ and $R_2=-I$, in types 3, 4 and 6, corresponding to square and hexagonal lattices
\beq{form}
R_3=\left(\begin{array}{cc} -\frac{1}{2}& -\frac{\sqrt 3}{2} \\  \frac{\sqrt 3}{2}& -\frac{1}{2} \end{array}\right),\, 
 R_4=\left(\begin{array}{cc} 0& -1 \\  1 & 0 \end{array}\right),\, 
R_6=\left(\begin{array}{cc} \frac{1}{2}& -\frac{\sqrt 3}{2} \\  \frac{\sqrt 3}{2}& \frac{1}{2} \end{array}\right).
\eeq
respectively. For $n=1$ we have the {\it Bianchi torus} $$M^3_1=T_{\mathcal L}^3=E(2)/\mathcal L,$$
where $\mathcal L \subset \mathbb R^2(q_1,q_2)$ is any two-dimensional lattice, with the corresponding left-invariant $E(2)$-metrics
\beq{met}
ds^2=A^{-1} dq_0^2+B^{-1}(dq_1^2+dq_2^2 +(B-1)(\sin q_0 dq_1+\cos q_0 dq_2)^2).
\eeq
In other cases $M^3_n=T_{\mathcal L}^3/\mathbb Z_n$ are simply quotients of the Bianchi torus under additional assumption of rotational symmetry $R$ of the lattice $\mathcal L.$
%(classically other types $M^3_k$ are different only by a finite cover).

\begin{prop}
The geodesic flow on Bianchi-VII$_0$ threefolds $M_n^3=T_{\mathcal L}^3/ G_n$ is integrable in Liouville sense with three polynomial in momenta Poisson commuting integrals.
\end{prop}

Indeed, on the quotient threefolds $M^3_n$ as 3 independent Poisson commuting integrals of the geodesic problem one can choose Hamiltonian (\ref{Ham}) and $F_1=p_1, F_2=p_2$ for the general case $n=1$ and $G_n$-invariant functions of momentum in the symmetric cases: $F_1=p_1^2, F_2=p_2^2$ for $n=2$,
$F_1=p_1^2+p_2^2, F_2=p_1^3-3p_1p_2^2$ for $n=3,$  $F_1=p_1^2+p_2^2, F_2=p_1^2p_2^2$ ($n=4$), $F_1=p_1^2+p_2^2, F_2=(p_1^3-3p_1p_2^2)^2$ ($n=6$).

To explain the behaviour of the corresponding geodesics, it is enough to consider the case of the Bianchi torus $M^3_1$. Let us assume first that the lattice $\mathcal L=2\pi \mathbb Z^2  \subset \mathbb R^2(q_1,q_2)$ and that the integral $p_2=c_2=0,$ but $p_1=c_1\neq 0$. Since $q_2$ is bounded, the corresponding geodesics will stay near the subtorus $T^2\subset T_{\mathcal L}^3$ with $q_2=0.$ 
To describe qualitatively the corresponding projections $(q_0, q_1)$  of the geodesics consider a periodic 
surface of revolution in cylindrical coordinates $(\rho, \varphi, z)$ in $\mathbb R^3$ having the form $\rho=f(z)$ with periodic $f(z)=a+b \sin^2z$ with positive $a,b$.

The geodesics on surfaces of revolutions satisfy the classical {\it Clairault's  relation}:
\beq{Clairault}
\mathcal C:=\rho \cos\theta=C
\eeq
is constant along the geodesics, where $\theta$ is the angle between the geodesic and parallel (see e.g. \cite{DC}).
In our case we have two special geodesics, which are parallels with maximal and minimal length (which we call equator and neck respectively).
We have 5 different cases: $$C=R,\,  r<C<R,\, C=r, \, 0<C<r, \, C=0$$ where $r=a$ and $R=a+b$ are radii of neck and equator respectively. 
The behaviour of the corresponding geodesics is shown on Fig 2. 
In the critical case $C=r$ we have closed neck geodesics and the geodesics spiralling towards the neck.

%\begin{figure}[h]
%  \includegraphics[width=30mm]{scattering1}\quad
%  \includegraphics[width=35mm]{scattering2}\quad
%  \includegraphics[width=37mm]{scattering3}
%  \caption{Geodesic scattering on one-sheeted hyperboloid}
%\end{figure}

\begin{figure}[h]
  \includegraphics[width=20mm]{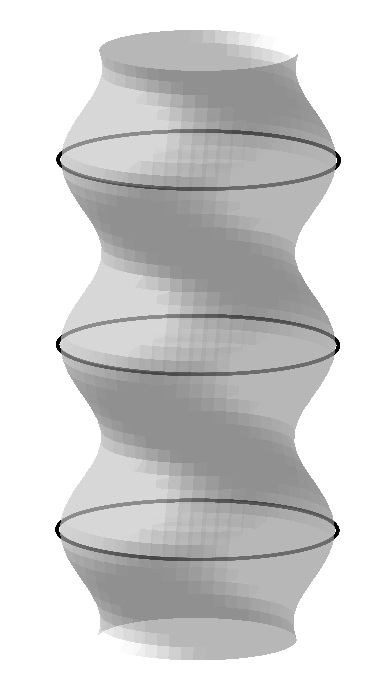} \quad \includegraphics[width=18.5mm]{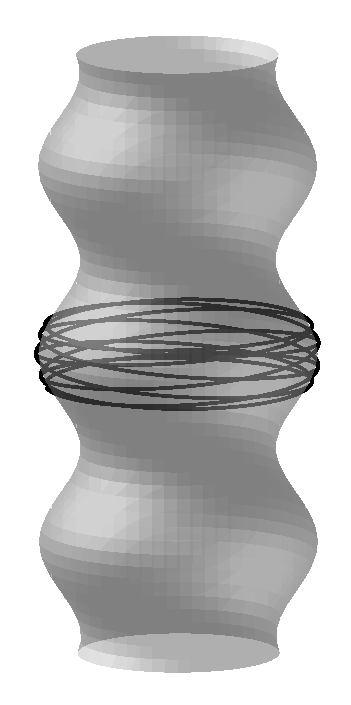} \quad   \includegraphics[width=19.5mm]{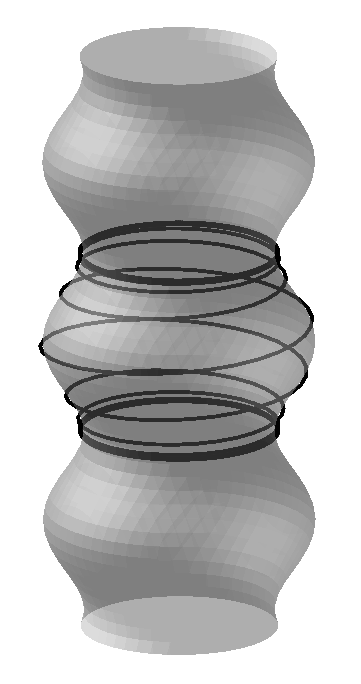}  \quad   \includegraphics[width=20mm]{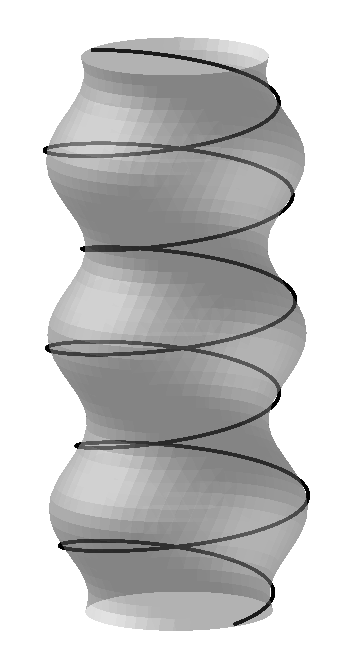}  \quad   \includegraphics[width=19mm]{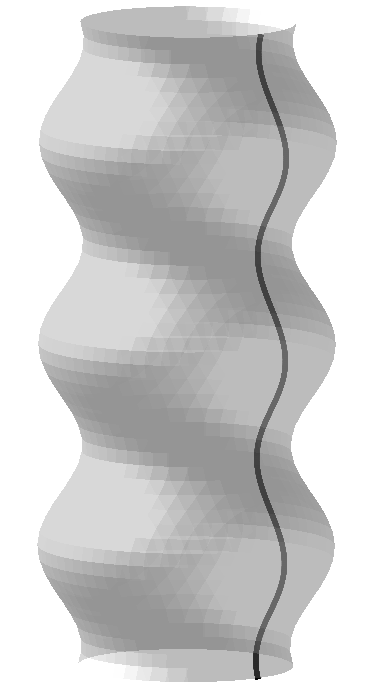} 
  \caption{Geodesics on periodic surface of revolution}
\end{figure}

If we put here $q_0=z\, (mod \, 2\pi), \, q_1=\varphi$, then we get a qualitative description of the motion of the projected geodesics $(q_0,q_1)$ on the torus $T^2.$
More precisely, the first four cases above correspond to the four cases $$2H=p^2, \,\, p^2<2H<Bp^2, \,\, 2H=Bp^2, \,\, 2H>Bp^2,$$ where we assume that $p^2:=p_1^2+p_2^2>0$ and  $B>1$. When $p^2=0$ (meaning that $p_1=p_2=0$) we have the last (meridian) case with $C=0,$ corresponding to geodesics with constant $q_1, q_2$ and uniformly changing $q_0.$

Interestingly enough, in the $Sol$-case \cite{BT} it was the analogue of the last family of simplest geodesics on $M_A^3$, which led to the chaos and positive topological entropy because of the Anosov gluing map $A: T^2\to T^2.$ In our case the gluing map $R_n$ is simply a rotation of order $n,$ which creates no problems for integrability.
 
 For a generic lattice $\mathcal L$ the plane $q_2=0$ is embedded in the Bianchi torus quasi-periodically, so for generic values of integrals with $2H>B(p_1^2+p_2^2)$ the geodesics are quasi-periodic and dense on the torus (similarly to the flat case).

Note that for $B=1$ we do have a flat metric, corresponding to the 3D Euclidean geometry in Thurston's classification. 
Bianchi geometry (\ref{met}) can be viewed as special deformation of this geometry with the smaller symmetry group $E(2)\subset E(3)$, which is not maximal and thus is missing in Thurston's classification. %In that case corresponding surface of revolution is usual cylinder and we have only the last case of winding lines on it.

%Depending on the position of the vector $(p_1=c_1, p_2=c_2)$ with respect to the lattice $\mathcal L$ the corresponding geodesics are periodic, or quasi-periodic on Bianchi torus $T_{\mathcal L}^3$, similarly to the Euclidean torus case. 

We consider now the quantum situation, which more essentially depends on the geometry of the lattice $\mathcal L$ (as we know already in the flat torus case).

\section{Spectral problem on Bianchi-VII$_0$ threefolds}

Consider first the case of Bianchi torus $T_\mathcal L^3$ with metric (\ref{met}). 
Let $$v_1=\omega_{1}^{1}q_1+\omega_{1}^{2}q_2, \, v_2= \omega_{2}^{1}q_1+\omega_{2}^{2}q_2$$ be a basis in the corresponding dual lattice $2\pi \mathcal L^*$:
$$
2\pi \mathcal L^*=\{k_1(\omega_{1}^{1}q_1+\omega_{1}^{2}q_2)+k_2(\omega_{2}^{1}q_1+\omega_{2}^{2}q_2),\,\,\, (k_1,k_2) \in \mathbb Z^2 \}.
$$
By definition, the elements of $\mathcal L^*$ take integer values on the lattice $\mathcal L,$ so 
\beq{integ}
(k, \Omega q):=k_1(\omega_{1}^{1}q_1+\omega_{1}^{2}q_2)+k_2(\omega_{2}^{1}q_1+\omega_{2}^{2}q_2) \in 2\pi \mathbb Z
\eeq
 for all $(k_1,k_2) \in \mathbb Z^2$ and $(q_1,q_2)\in \mathcal L.$

The corresponding Laplace-Beltrami operator on $T_\mathcal L^3$ has the form
\beq{LB}
\Delta = A\frac{\partial^2}{\partial q_0^2}+\frac{\partial^2}{\partial q_1^2}+\frac{\partial^2}{\partial q_2^2}+(B-1)\left(\cos q_0 \frac{\partial}{\partial q_1}-\sin q_0 \frac{\partial}{\partial q_2}\right)^2.
\eeq
This is a self-adjoint operator in the Hilbert space $L_2(T_\mathcal L^3)$ of functions on the torus.
The  spectral geodesic problem is to study the spectrum of $\Delta$ in this space:
\beq{spec}
-\Delta \Psi=\mathcal E \Psi.
\eeq
Because the coefficients of $\Delta$ depend only on $q_0$, we separate variables and look for the eigenfunctions of $\Delta$ of the form
\beq{Psi}
\Psi_k(q)=e^{i (k,\Omega q)} \psi(q_0), \,\, k=(k_1,k_2) \in \mathbb Z^2.
\eeq
The condition (\ref{integ}) guarantees that the corresponding function is well-defined in the Bianchi torus.

Substituting this into  (\ref{spec}) and denoting $x=q_0$ we see that the $\psi(x)$ must satisfy the equation

$$
-A\frac{\dee^2 \psi}{\dee x^2} +[Q(k)+(B-1)((\omega_{1}^{1}k_1+\omega_{2}^{1}k_2)\cos x-(\omega_{1}^{2}k_1+\omega_{2}^{2}k_2)\sin x)^2] \psi(x)=\mathcal E\psi(x),
$$
where 
\beq{Q}
Q(k)=(\omega_{1}^{1}k_1+\omega_{2}^{1}k_2)^2+(\omega_{1}^{2}k_1+\omega_{2}^{2}k_2)^2.
\eeq
Define the angle $\alpha=\alpha(k)$ by the relations $$\omega_{1}^{1}k_1+\omega_{2}^{1}k_2=\sqrt{Q(k)}\cos \alpha,\\ \,\, \omega_{1}^{2}k_1+\omega_{2}^{2}k_2=\sqrt{Q(k)}\sin \alpha,$$ then this equation will become the classical Mathieu equation (shifted by $\alpha$)
\beq{Math}
-\frac{\dee^2 \psi}{\dee x^2} +\mu \cos 2(x+\alpha) \psi(x)=\lambda\psi(x),
\eeq
where
\beq{mu}
\mu:=\frac{(B-1)Q(k)}{2A}, \,\, \lambda:=\frac{2\mathcal E-(B+1)Q(k)}{2A}.
\eeq

\begin{prop}
A function $\Psi$ of the form (\ref{Psi}) is an eigenfunction of the Laplace-Beltrami operator (\ref{LB}) if and only if $\psi(x)$ satisfies the
Mathieu equation (\ref{Math}) with parameters (\ref{mu}).
%\beq{separ}
%\left(-\frac{\dee^2}{\dee x^2}  + (B-1)K^2 \cos^2(x+\alpha (k))\right)\psi(x)=\lambda \psi (x)
%\eeq
\end{prop}

Recall that the classical {\it Mathieu equation} \cite{Mathi} has the form
\beq{clasmat}
\frac{\dee^2 y}{\dee x^2}  +  (\lambda-\mu \cos 2 x) y  = 0.
\eeq Its solutions are known as {\it Mathieu functions}, see  Whittaker-Watson ~\cite{WW} and, for historical perspective, recent  review \cite{BCZ}. We should mention that our notation of the parameters $\lambda$ and $\mu$ are different from $a$ and $q$ (used by Whittaker and Watson \cite{WW}) and $h$ and $\theta$ (used by Erdelyi \cite{Er}).

We will be interested only in $2\pi$-periodic Mathieu functions, which are also eigenfunctions of the Mathieu operator 
\beq{Loper}
L=-\frac{\dee^2}{\dee x^2}+ \mu \cos 2x
\eeq
on the circle $S^1=\mathbb R/2\pi i \mathbb Z.$
For given $\mu$ there are infinite number of values of $\lambda=\lambda_l(\mu), \, l=0,1,\dots$, for which Mathieu equation has periodic solutions, denoted in a properly normalised form as $ce_m(x,\mu), se_m(x,\mu)$, see \cite{BCZ, Er}. 

When $\mu=0$ the spectrum is degenerate:  $\lambda=m^2, \, m=0,1,\dots$  with two-dimensional space of eigenfunctions generated by $\cos mx, \, \sin mx$ (except for $m=0$ corresponding to the ground state $\varphi\equiv 1$). 
The eigenfunctions $ce_m(x,\mu), se_m(x,\mu)$ with the eigenvalues $a_m$ and $b_m$ respectively are analogues of $\cos mx$ and $\sin mx$ for $\mu\neq 0$ having the following properties:
$$
ce_m(-x,\mu)=ce_m(x,\mu), \quad se_m(-x,\mu)=-se_m(x,\mu),
$$
$$
ce_m(x+\pi,\mu)=(-1)^m ce_m(x,\mu), \quad se_m(x+\pi,\mu)=(-1)^m se_m(x,\mu).
$$

\begin{figure}[h]
\begin{center} 
\includegraphics[height=52mm]{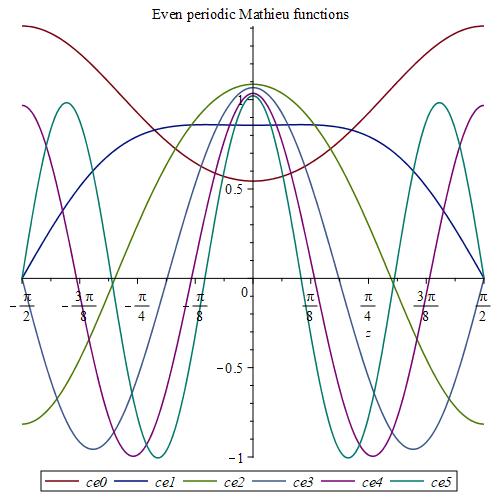}  \hspace{8pt}  \includegraphics[height=52mm]{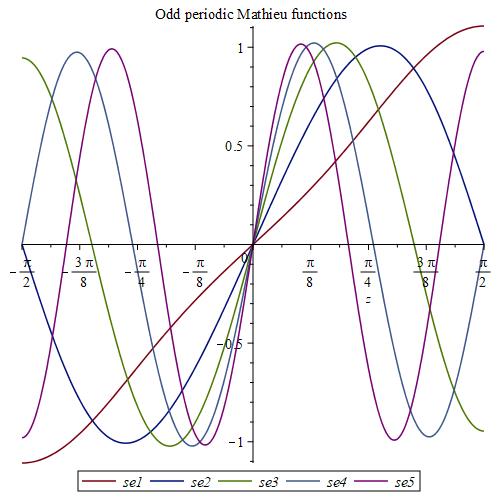}
\caption{\small First Mathieu functions $ce_m(x,\mu), se_m(x,\mu)$ for $\mu=1.$}
\end{center}
\end{figure}

For $\mu<0$ the eigenvalues $a_m=a_m(\mu), \, b_m=b_m(\mu)$ satisfy the inequalities
$$
a_0<b_1<a_1<b_2<a_2<b_3<a_3\dots,
$$
while for $\mu>0$ we have
$$
a_0<a_1<b_1<b_2<a_2<a_3<b_3<\dots
$$
(see \cite{Er}). Asymptotically for large $m$ $a_m(\mu)$ and $b_m(\mu)$ are exponentially close to $m^2.$

We denote these eigenvalues in the increasing order as $\lambda_l(\mu), \, l\in \mathbb Z_{\geq 0}$ and the eigenfunctions respectively as $\varphi_l(x,\mu):$ 
\beq{Lphi}
L \varphi_l(x,\mu)=\lambda_l(\mu) \varphi_l(x,\mu).
\eeq

%For $\mu\neq 0$ the spectrum is known to be simple with $\lambda_n$ asymptotically behaving as $$\lambda_n \sim \left[\frac{n+1}{2}\right]^2, \quad n\to\infty,$$
%where $[a]$ denotes the integer part of $a\in \mathbb R.$

Let
\beq{eigenf}
\Psi_{k,l}(q):=e^{i (k,\Omega q)} \varphi_l(q_0,\mu)
\eeq
be the corresponding eigenfunctions of the Laplace-Beltrami operator on Bianchi torus $T_\mathcal L^3$ labelled by $k=(k_1,k_2)\in \mathbb Z^2_\times:=\mathbb Z^2\setminus {(0,0)}$ and $n\in \mathbb Z_{\geq 0}$.
We consider separately the eigenfunctions related to $k=(0,0)$
\beq{eigenf0}
\Psi_{0,m} (q):=e^{im q_0}, \quad m \in {\mathbb Z}
\eeq
with the eigenvalues $\lambda_m = m^2.$

\begin{theorem}
The eigenfunctions (\ref{eigenf}),(\ref{eigenf0}) of the Laplace-Beltrami operator (\ref{LB})
 form a complete basis in the Hilbert space 
$L_2(T_\mathcal L^3)$ of functions on Bianchi torus $T_\mathcal L^3.$
For $B\neq 1$ the corresponding spectrum $Spec= Spec_0 \cup Spec_M$ consists of two parts:
\beq{spec0}
Spec_0 =\{
\mathcal E_m = A m^2 ,\quad m \in \mathbb Z_{\geq 0} \},
\eeq
\beq{specm}
Spec_M = \{\mathcal E_{k, l} = A\lambda_l(\mu)+\frac{(B+1)Q(k)}{2}, \, \, k\in \mathbb Z^2_\times, \, l \in \mathbb Z_{\geq 0}\},
\eeq
where $\lambda_l$ are the eigenvalues of the Mathieu operator (\ref{Loper}) with $\mu=\frac{(B-1)Q(k)}{2A},$ $Q(k)$ is the quadratic form of the lattice $\mathcal L.$
\end{theorem}

\begin{proof}
We need to show the completeness of the set of our eigenfunctions. For this it is enough to prove that any smooth function $\Phi: M^3\to
\mathbb R$ which is orthogonal to each of these eigenfunctions
is zero.

Denoting again $q_0=x$ expand the function $\Phi(x, q_1, q_2)$ as a Fourier series with respect to $q_1,q_2:$
$$
\Phi(x, q_1, q_2)=\sum_{s=(s_1,s_2)\in \mathbb Z^2} a_s(x) e^{i(s,\Omega q)}
$$
with some smooth $2\pi$-periodic coefficients $a_s(x), \, x \in \mathbb R.$
Orthogonality of $\Phi(x, q_1, q_2)$ to all eigenfunctions $\Psi_{k,l}(q)$ implies that for all non-zero $s\in \mathbb Z^2$ Fourier coefficient $a_s(x)$ is orthogonal to all Mathieu functions $\varphi_l(x,\mu)$ for all $l\in \mathbb Z_{\geq 0}.$ Since the system of the Mathieu functions is complete in the Hilbert space of $2\pi$-periodic functions  this implies that the corresponding $a_s(x)\equiv 0$ and thus $\Phi=a_0(x)$.
Now orthogonality to $\Psi_{0,l} (q)$ means that $a_0(x)$ must vanish and thus $\Phi \equiv 0.$
\end{proof}

\begin{cor}
The spectrum of the Bianchi torus with metric (\ref{met}) behaves asymptotically as the values on 3D integer lattice of the quadratic form
\beq{tilQ}
\tilde Q(k_0,k_1,k_2)= Ak_0^2+\frac{(B+1)Q(k_1,k_2)}{2}, \quad (k_0,k_1,k_2) \in \mathbb Z^3
\eeq
with quadratic form $Q(k_1,k_2)$ given by (\ref{Q}).
\end{cor}

Indeed, it is known that the Mathieu spectrum behaves asymptotically as $m^2, \, m \in \mathbb Z.$ When $B=1$ we have the flat torus with the spectrum given by the values of the corresponding form $\tilde Q=Ak_0+Q(k_1,k_2)$, so in that case the result is exact. The modification (\ref{tilQ}) of this form for general  $B$ is non-obvious and seems to be quite interesting.

Consider now the case of the quotient of Bianchi torus $M_2^3=T_\mathcal L^3/\mathbb Z_2$ by the involution $\tau: (q_0,q_1,q_2)\to (q_0+\pi, -q_1,-q_2).$
Corresponding eigenfunctions can be found by the averaging:
\beq{eigenf2}
\Psi^{(2)}_{k,l}(q):=e^{i (k, \Omega q)} \varphi_l(q_0,\mu)+e^{-i (k,\Omega q)} \varphi_l(q_0+\pi,\mu)
\eeq

More explicitly, for $\mu<0$ (corresponding to $B<1$) we have the functions
$$
\varPhi^{(2)}_{k, 4j}=\cos (k,\Omega q) ce_{2j}(x,\mu),\, \varPhi^{(2)}_{k, 4j+1}=\sin (k,\Omega q) se_{2j+1}(x,\mu),$$
$$
\varPhi^{(2)}_{k, 4j+2}=\sin (k,\Omega q) ce_{2j+1}(x,\mu),\, \varPhi^{(2)}_{k, 4j+3}=\cos (k,\Omega q) se_{2j+2}(x,\mu),
$$
where $k \in \mathbb Z^2_\times/\pm$ is considered up to a sign and $j \in \mathbb Z_{\geq 0}$.

The corresponding spectrum is given by the same formula (\ref{specm}), although the multiplicities for a generic lattice drop from 2 to 1 (since $k$ and $-k$ are now identified). For non-generic lattices the multiplicities could be higher, see next section.

The trivial part is to be modified to $Spec^{(2)}_0 =\{
\mathcal E_m = 4Am^2 ,\quad m \in \mathbb Z_{\geq 0} \},$
with the multiplicities 2 (except ground state $\mathcal E=0$ having multiplicity 1) and the eigenfunctions
$e^{\pm 2im q_0}.$

The remaining cases correspond to the symmetric lattices and have nice arithmetic properties. 

\section{Arithmetic of the symmetric Bianchi-VII$_0$ threefolds}

 Consider now the symmetric cases of Bianchi tori with the lattices having additional rotational symmetries of order $n>2.$
 When $n=4$ we have a square lattice, which can be assumed to be the standard square lattice $\mathcal L=2\pi \mathbb Z^2 \subset \mathbb R^2(q_1,q_2)$
with quadratic form
$
 Q(k)=k_1^2+k_2^2.
$
 
 The same calculations as before lead to the eigenfunctions
\beq{eigenf4}
\Psi_{k,l}(q):=e^{i (k_1q_1+k_2 q_2)} \varphi_l(\mu, q_0)
\eeq
be the corresponding eigenfunctions of the Laplace-Beltrami operator on $M^3$ labelled by $k=(k_1,k_2)\in \mathbb Z^2_\times:=\mathbb Z^2\setminus {(0,0)}$ and $l\in \mathbb Z_{\geq 0}$.
Together with the elementary eigenfunctions
$
\Psi_{0,m} (q):=e^{im q_0}, \, m \in {\mathbb Z}
$
with the eigenvalues $\lambda_m = A m^2$, they form a complete basis in the $L_2$ Hilbert space on the corresponding Bianchi torus.
The nontrivial part of the spectrum (\ref{specm}) has the form
\beq{nont}
\mathcal E_{K, l} = A\lambda_l(\mu)+(B+1)K/2,
\eeq
where  $K=k_1^2+k_2^2, \,\,  l \in \mathbb Z_{\geq 0}$ and $\lambda_l(\mu)$ are the eigenvalues of Mathieu operator (\ref{Loper}) with $\mu=\frac{(B-1)K}{2A}.$
Asymptotically the spectrum behaves like the values of the quadratic form
$$
\tilde Q=Ak_0^2+(B+1)(k_1^2+k_2^2)/2, \,\, (k_0,k_1,k_2)\in \mathbb Z^3.
$$

The difference with the general case is that now we have high multiplicities. Indeed, a given integer $K=k_1^2+k_2^2$ may be represented as sum of the squares
in many ways. This classical problem was studied by many mathematicians, including Fermat, Euler and Gauss.
Jacobi found the following formula for the number $N(K)$ of such representations (see e.g. \cite{Dick}):
\beq{gauss}
N(K) = 4(N_1(K)
- N_3(K)),
\eeq
where $N_1(K)$ and $N_3(K)$ are the numbers of the
divisors of $K$ with the residues $1$ and $3$ modulo 4
respectively. Another classical result due to E. Landau \cite{Landau} claims that the number of
integers $K\leq N$ represented as the sum of two squares grows as
$N/\sqrt{\log N}.$

Similar result holds for the Bianchi torus with the regular hexagonal lattice (covering  $M^3_n$ with $n=3$ and $n=6$) 
when the non-trivial eigenvalues have the form (\ref{nont}) with $K=k_1^2+k_1k_2+k_2^2, \, (k_1,k_2)\in \mathbb Z^2.$
In that case the number of representations $N(K)$ of a given integer $K$ by this form
is also well-known:
\beq{gauss2}
N(K) = 6(N_1(K)
- N_2(K)),
\eeq
where $N_1(K)$ and $N_{2}(K)$ are the numbers of the
divisors of $K$ with the residues $1$ and $2$ modulo 3
respectively (see \cite{Dick}). The corresponding generalisation of Landau result was proved by Pall \cite{Pall}.

\begin{theorem}
For generic parameters $A$ and $B$ in the metric the multiplicity of the eigenvalues (\ref{nont})
of the Laplace-Beltrami operator on the symmetric Bianchi-VII$_0$ threefolds $M_n^3$ with $n=3,4,6$
are
$$
m(\mathcal E_{K,l})=\frac{N(K)}{n},
$$
where $N(K)$ are given by (\ref{gauss}), (\ref{gauss2}) respectively.
%In particular, this means that the degeneracy of these eigenvalues grow as .
\end{theorem}

Indeed, the eigenfunctions of the Laplace-Beltrami operator on the corresponding Bianchi threefolds $M_n^3=T_\mathcal L^3/\mathbb Z_n$
can be found by averaging the eigenfunctions on the corresponding Bianchi torus $T_\mathcal L$ by the group generated by
$
\tau: (q_0, q_1, q_2) \to (q_0+2\pi/n, R_n(q_1,q_2))
$
with the rotation $R_n$ given by (\ref{form}).
The trivial part of the spectrum consists of the eigenvalues
$\mathcal E_m = A m^2n^2 ,\quad m \in \mathbb Z_{\geq 0}$
with the multiplicities 2 (for $m>0$), corresponding to the eigenfunctions
$e^{\pm 2imn q_0}.$

As a corollary of the number-theoretic results of Landau and Pall \cite{Landau, Pall} we have now the following result (cf. \cite{BDV}).

\begin{cor}
The level spacing distribution for the spectrum of the symmetric Bianchi tori and threefolds $M_n^3, \, n=3,4,6$ is not Poisson
and hence the Berry-Tabor conjecture does not hold in this case. 
\end{cor}

Recall that the Berry-Tabor conjecture \cite{Berry-Tabor} predicts that the quantisation of completely integrable Hamiltonian systems should have Poisson distributed level spacing, which means their consecutive spacings should be independent and identically distributed.
In our case the growing degeneracy means that most of level spacings are zero, which contradicts to Berry-Tabor conjecture. 

\section{Comparison with Bianchi-VI and Bianchi-IX cases}

The classical $Sol$-case (corresponding to Bianchi-VI$_0$ type) was studied by Bolsinov and Taimanov in \cite{BT}.
The principal examples here are the torus $T^2$-fibrations $M^3_A$ over $S^1$ with hyperbolic gluing maps $A:T^2 \rightarrow T^2, \, A\in SL(2,\mathbb{Z}).$
Bolsinov and Taimanov proved a surprising fact that the corresponding geodesic flow is Liouville integrable in smooth category, but not in analytic one.
The obstruction to the analytic integrability  \cite{T}  is purely topological and related to the exponential growth of the fundamental group $\pi_1(M_A^3)$. \footnote{It is interesting that these topological manifolds were considered already in 1892 by Poincar\'e, who was interested in the examples of manifolds with large fundamental groups.}
%Moreover, at the degenerate level the system is fully chaotic Anosov map, so the system has positive topological entropy.

In the quantum $Sol$-case considered in \cite{BDV} the separation of variables leads to the $\cosh$-version of Mathieu equation 
%$$
%L=-\frac{\dee^2}{\dee x^2}+ \mu \cosh 2x,
%$$
with parameter $\mu$ depending on the values of integer binary quadratic form $Q$, which is determined by the topology of the manifold (see \cite{BDV}). As a result, we always have the high degeneracy of the spectrum.
As we have seen in our Bianchi-VII case, the degeneracy of the eigenvalues disappears for the generic Bianchi tori $M_1^3$.

Indeed, in that case we have the positive definite quadratic form $Q(k)$ with generic real coefficients. The study of values of quadratic forms at integer points is a classical topic in number theory, going back to Fermat.

In particular, Sarnak showed in \cite{Sarnak}  that the pair correlation of values does satisfy Berry-Tabor prediction on a set of full measure in the space of such forms.
Eskin, Margulis and Mozes \cite{Eskin} provided stronger result by giving explicit diophantine conditions on the coefficients of the form $Q(k)$ for which this holds.
%Recall that the number $\beta$ is called {\it diophantine} if there exist constants $k\geq 2$ and $C>0$ such that $|\beta -\frac{p}{q}|>Cq^{-k}$ for all integer $p,q>0$.
%Eskin et al \cite{Eskin}
In particular, they proved that if $\beta>0$ is diophantine then the pair correlation of 
the values of $Q(k)=k_1^2+\beta k_2^2, \, k_1,k_2 \in \mathbb Z$ agrees with the Berry-Tabor conjecture.

As we have shown, in the Bianchi-VII$_0$ case the spectrum asymptotically behaves as the values of the quadratic form (\ref{tilQ}) on 3D integer lattice,
%$$ \tilde Q= Ak_0^2+\frac{(B+1)Q(k_1,k_2)}{2}, \,\, (k_0,k_1,k_2) \in \mathbb Z^3,$$
so it is natural to expect that the Berry-Tabor conjecture is true for the generic Bianchi tori, but this is still to be justified.

Like in the $Sol$-case \cite{BDV}, we can visualise the quantum monodromy for the symmetric Bianchi threefolds $M_n^3$ with $n=3,4,6$ by plotting the values of the $\mathbb Z_n$-invariant integrals $F_1=\sqrt{Q(p)} \cos n\varphi(p), \, F_2=\sqrt{Q(p)} \sin n\varphi(p)$ with $p\in \mathcal L^*$ and $\varphi$ being its angular coordinate. In particular, for $n=3$ we have the grid shown on Fig. 3 against the corresponding ``Sol-flower" from \cite{BDV} with the demonstration of the quantum monodromy in that case. In the $Sol$-case the quantum monodromy  is given by the ``cat-map" $$A=\left(\begin{array}{cc} 2& 1 \\  1 & 1 \end{array}\right),$$ while in our case it is simply rotation $R_3$ by $2\pi/3.$

\begin{figure}[h]
  \includegraphics[width=48mm]{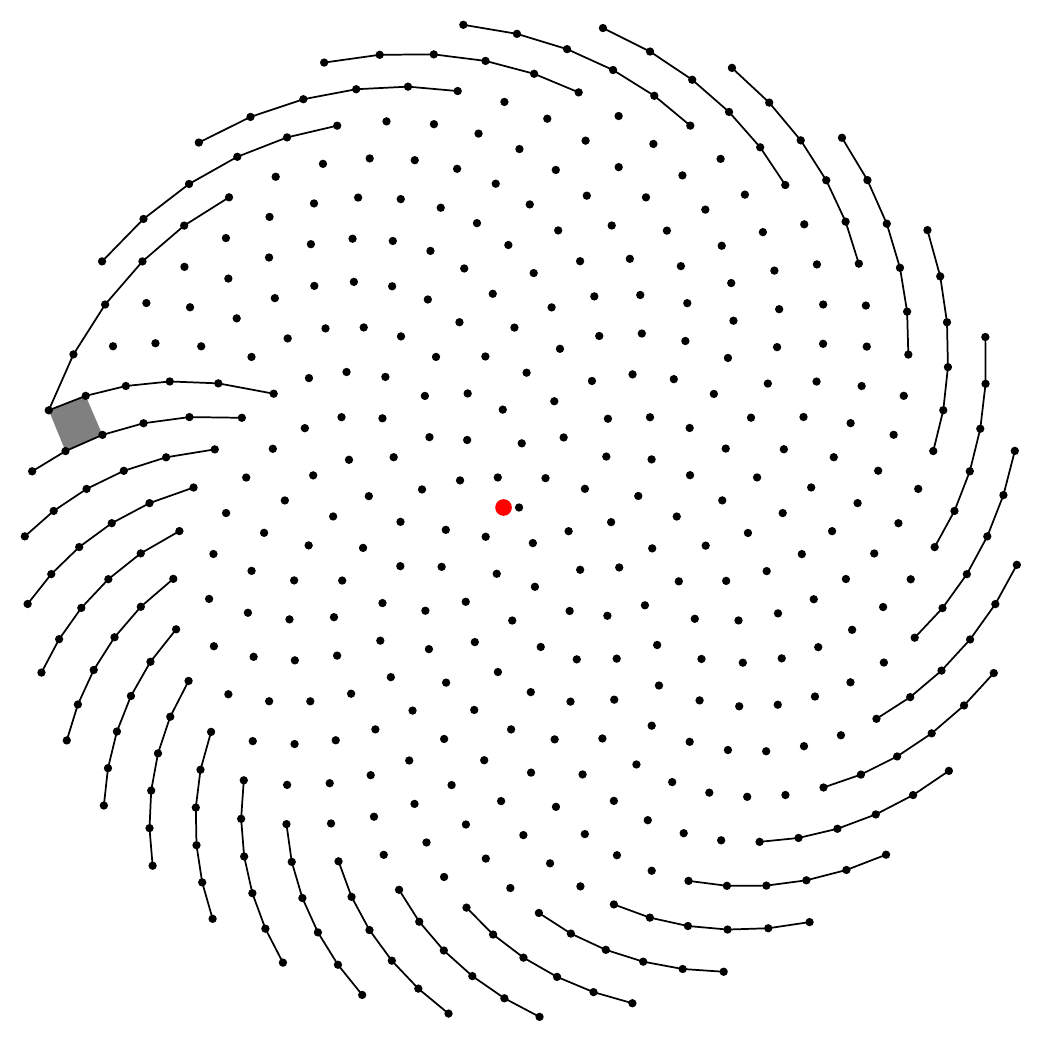} \quad \includegraphics[width=57mm]{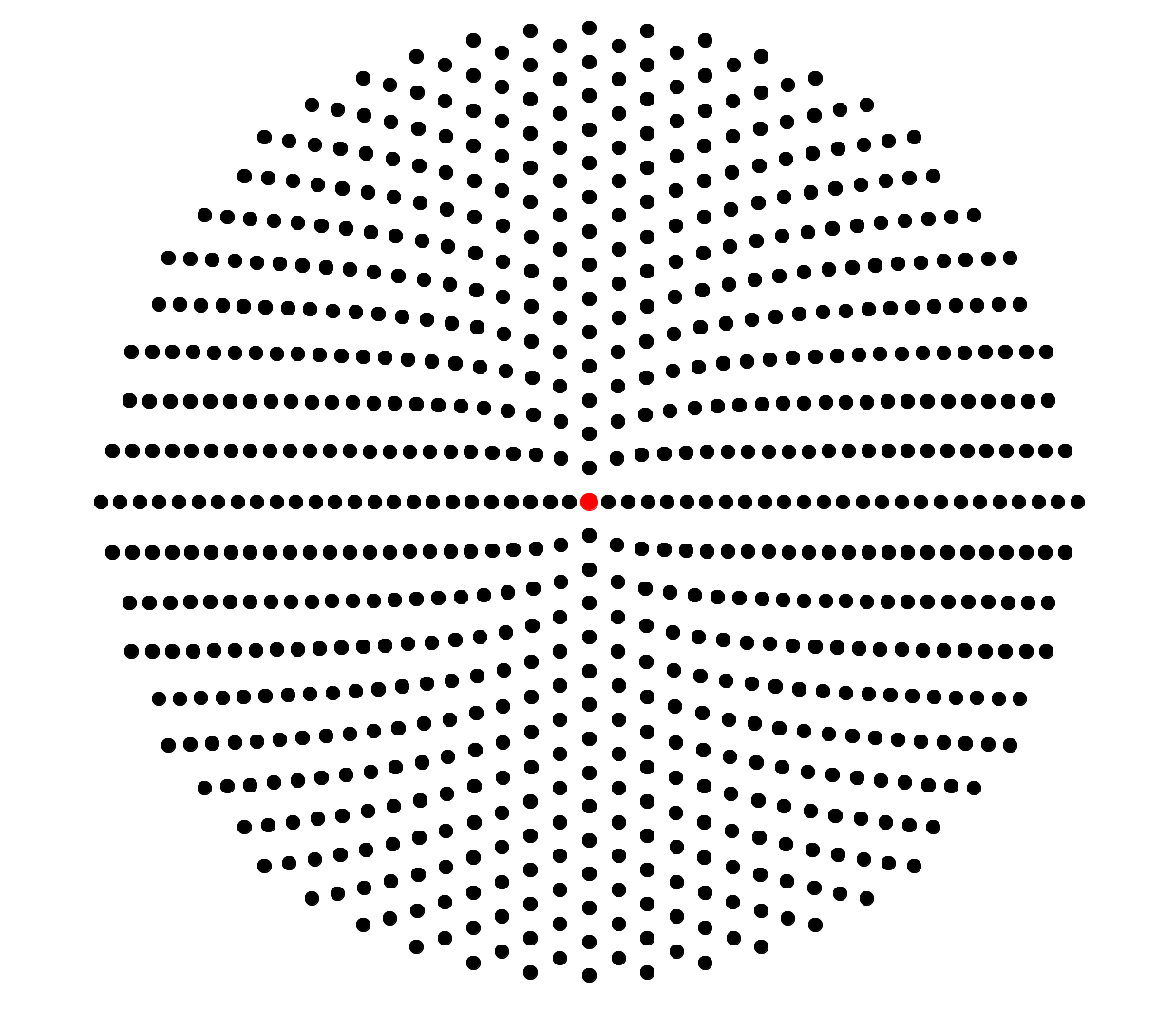} 
  \caption{Quantum monodromy in $Sol$-case and corresponding grid for Bianchi-VII threefold $M_3^3.$}
\end{figure}

Finally, let us compare our case with the Bianchi-IX case. In that case we have Lie group $SO(3)$ (or, its double cover $SU(2)$) with a left-invariant metric and the geodesics describing the motion of the classical Euler top \cite{Arnold}, which were found explicitly in terms of elliptic functions by Jacobi \cite{Jacobi}. 

In quantum case the separation of variables in spherical elliptic coordinates \cite{KI, GV} leads to the classical Lam\`e equation
$$
L=-\frac{\dee^2}{\dee x^2}+ m(m+1)k^2 \,sn^2(x,k), \quad m\in \mathbb Z_{\geq 0},
$$
where $sn(x,k)$ is the classical Jacobi elliptic function \cite{WW}. Note that the set of parameters $m(m+1), m \in \mathbb Z_{\geq 0}$ is the spectrum of the unit sphere $S^2$, while in our Bianchi-VII$_0$ case we have the spectrum of a flat 2D torus. %A possible explanation is that $SU(2)=S^3$ is the Hopf fibration over $S^2$ with the fibre $S^1.$ 
The corresponding Lam\`e operators are known to have remarkable spectral property: in their spectrum there are only finitely many (namely, exactly $m$) gaps \cite{Ince}. The ends of the gaps can be computed using classical methods going back to Hermite and Halphen \cite{WW} (see also \cite{GV} and references therein), but the formulae become very complicated when $m$ grows. It is interesting to note that if $m\to \infty, \, k\to 0$ such that $mk\to \nu,$ then
$$m(m+1)k^2 \,sn^2(x,k) \to \nu^2\sin^2x,$$ and we have as a limit a (shifted) Mathieu operator from Bianchi-VII$_0$ case.

In the most symmetric Bianchi-IX case we have the unit sphere $S^3$ with the spectrum $m(m+2), \, m \in \mathbb Z_{\geq 0}$ with the multiplicity $(m+1)^2$, while in Bianchi-VII$_0$ case we have flat torus $T^3$ with spectrum given by the values of a quadratic form on the corresponding 3D lattice. In Thurston's list these two cases are represented by $S^3$ and $E^3$ respectively. 

Note that in Bianchi-IX case we have also the quotients $SU(2)/G$, where $G$ is either cyclic group $\mathbb Z_n$, or the double cover of one of the finite subgroups of $SO(3)$:
dihedral  $D_n,$ tetrahedral $T$, octahedral $O$ and icosahedral $I$ (in the last case the corresponding quotient is famous Poincar\'e homology sphere). The spectrum of the corresponding symmetric versions can be found in \cite{Ikeda}.

%Note that for large $n$ the Mathieu eigenvalues behave asymptotically as $\lambda_n \sim n^2$, so the large Bianchi eigenvalues $\mathcal E$ behave approximately like the values at the integer points of quadratic forms in 3 variables.

%\medskip

\section{Acknowledgements}
We are very grateful to Alexey Bolsinov for many helpful discussions. 

The work of the second author (Y.Ye) was supported by the XJTLU Research Development Funding (grant number: RDF-21-01-031).


\begin{thebibliography}{99}

\bibitem{Arnold} 
V.I. Arnold {\it Mathematical Methods of Classical Mechanics.} Springer, 1984.

\bibitem{Bianchi}
L. Bianchi {\it Sugli spazii a tre dimensioni che ammettono un gruppo continuo di movimenti.} Soc. Ital. Sci. Mem. di Mat. {\bf 11} (1898), 267.

\bibitem{Berry-Tabor} M.V. Berry and M. Tabor {\it Level clustering in the regular spectrum.}
Proc. Roy. Soc. London A {\bf 356} (1977), 375-394.

\bibitem{BB} A. Bolsinov, J. Bao {\it  Note about integrable systems on low-dimensional Lie groups and Lie algebras.}
Regular and Chaotic Dynamics {\bf 24} (2019), 266--280.

\bibitem{BDV}
A.V. Bolsinov, H.R. Dullin, A.P. Veselov  {\it  Spectra of Sol-manifolds: arithmetic and  quantum monodromy.}  Comm. Math. Phys. {\bf 264} (2006), 583-611.

\bibitem{BT} 
A.V. Bolsinov, I.A. Taimanov {\it Integrable geodesic flows with positive topological entropy.}
Invent. Math. {\bf 140} (2000), 639-650.

\bibitem{BVY} 
A.V. Bolsinov, A.P. Veselov, Y. Ye {\it Chaos and integrability in $SL(2,\mathbb R)$-geometry.} Russian Math. Surveys {\bf 76(4)} (2021), 557-586.

\bibitem{BCZ}
C. Brimacombe, R.M. Corless, M. Zamir {\it Computation and applications of Mathieu functions: A historical perspective.}  SIAM Review {\bf 63:4} (2021), 653--720.

\bibitem{Butler}
L. Butler {\it A new class of homogeneous manifolds with Liouville-integrable geodesic flows.} C. R. Math. Acad. Sci. Soc. R. Can. {\bf 21}(1999), no. 4, 127-131. 

\bibitem{Coxeter} H.S.M. Coxeter {\it  Introduction to Geometry} (2nd ed.), Wiley, 1989.

\bibitem{DC}
M.P. do Carmo {\it Differential Geometry of Curves and Surfaces.} Prentice-Hall, 1976.

\bibitem{Dick}
L.E. Dickson {\it Introduction to the Theory of Numbers.} Univ. of Chicago Press, 1929.

\bibitem{Er}
A.Erdelyi (Editor) {\it Higher Transcendental Functions.},
Vol.1-3. McGraw-Hill, 1953.

\bibitem{Eskin}
A. Eskin, G. Margulis and S. Mozes {\it Quadratic forms of signature (2,2) and eigenvalue spacings on rectangular 2-tori.} Annals of Mathematics, {\bf 161} (2005), 679--725.

\bibitem{Fagundes}
H.V. Fagundes {\it Relativistic cosmologies with closed, locally homogeneous spatial sections.} Phys. Rev. Lett. {\bf 54} (1985), 1200-1202.

\bibitem{FIK}
Y. Fujiwara, H. Ishihara, H. Kodama {\it Comments on closed Bianchi models.} Class. Quantum Grav. {\bf 10} (1993), 859--867. 

\bibitem{GV} M.-P. Grosset, A.P. Veselov {\it Lame equation, quantum Euler top and elliptic Bernoulli 
polynomials.} Proc. Edinb. Math. Soc. {\bf 51} (2008), 635-650.

\bibitem{Ikeda}
A. Ikeda {\it On the spectrum of homogeneous spherical space forms.} Kodai Math. J. {\bf 18} (1995), 57-67.

\bibitem{Ince}
E.L. Ince {\it Further investigations into the periodic Lam\'e functions.}  Proc. Roy. Soc. Edinburgh 60, (1940). 83--99. 

%\bibitem{Hua} H.L. Keng {\it Introduction to Number Theory.}
%Springer Verlag, 1982.

\bibitem{Jacobi}
C.G.J. Jacobi {\it Sur la rotation d'un corps.} In: Gesammelte Werke, 2.Band, 291-352, 1881.


%\bibitem{James}
%R.D. James {\it The distribution of integers represented by quadratic forms.}
%Amer. J. Math. {\bf 60:3} (1938), 737-744

\bibitem{KI} H.A.~Kramers and G.P.~Ittmann   {\it Zur Quantelung des asymmetrischen Kreisels. I, II.} Z. Physik 53, 553-565, 58, 217-231 (1929), 217-231.
%; {\it Zur Quantelung der asymmetrischen Kreisels, II.} Z. Physik 58, 217-231 (1929).

\bibitem{Landau} %E. Landau {\it Elementary Number Theory.} Chelsea Publ. Co., New York, 1958.
E. Landau {\it \"Uber die Einteilung der
positive ganzen Zahlen in vier Klassen.}  Archiv der Mathematik und Physik (3), {\bf 13} (1908), 305-312.

\bibitem{LeVeque}
W.J. LeVeque {\it Fundamentals of Number Theory.} Dover Publ. inc., 1996.

\bibitem{Mathi}
E. Mathieu {\it M\'emoire sur le mouvement vibratoire d'une membrane de forme elliptique.} J. Math. Pure Appl. {\bf 13} (1868), 137--203.

\bibitem{Milnor}
J. Milnor {\it Curvatures of left invariant metrics on Lie groups.} Adv. Math. {\bf 21:3} (1976), 293-329.


\bibitem{Pall}
G. Pall {\it The distribution of integers represented by binary quadratic forms.}
Bull. Amer. Math. Soc. {\bf 49(6)} (1943), 447-449.

\bibitem{Sarnak} P. Sarnak {\it Values at integers of binary quadratic forms.}
CMS Conf. Proc., 21,  Amer. Math. Soc., Providence, RI, (1997), 181-203.

\bibitem{Scott} P. Scott {\it The geometries of 3-manifolds.} Bull. London Math. Soc., {\bf 15} (1983), 401-487.

\bibitem{T} I.A. Taimanov {\it Topological obstructions to
integrability of geodesic flows on non-simply-connected manifolds.}
Math. USSR Izv. {\bf 30} (1988), 403-409.

\bibitem{Th} W.P. Thurston {\it Hyperbolic geometry and $3$-manifolds. }  
LMS Lecture Notes Series {\bf 48}, Camb. Univ. Press, 1982. 

\bibitem{WW} 
E.T. Whittaker, G.N. Watson {\it Course in Modern Analysis}, Camb. Univ. Press (1990).


\end{thebibliography}
\end{document}